\documentclass[journal]{IEEEtran}
\usepackage{amssymb}
\usepackage{amsmath}
\usepackage{graphicx}
\usepackage{cite}

\usepackage{color}
\usepackage{citesort}
\usepackage{multicol}
\usepackage{amsfonts}
\usepackage{bbm}
\usepackage{subfigure}
\usepackage{graphicx,epstopdf}
\usepackage{epsfig}	
\usepackage{cite,graphicx,amsmath,amssymb}
\usepackage{comment}
\usepackage{amsmath}
\usepackage{lipsum}
\usepackage{stfloats}
\usepackage{algorithm}
\usepackage{algorithmic}

\newlength{\halfpagewidth}
\divide\halfpagewidth by 2

\newtheorem{theorem}{\textbf{Theorem}}

\newtheorem{lemma}{\textbf{Lemma}}

\newtheorem{corollary}{\textbf{Corollary}}

\newtheorem{proof}{\textbf{Proof}}

\newtheorem{definition}{\textbf{Definition}}

\newtheorem{proposition}{\text{Proposition}}

\makeatletter
\def\ScaleIfNeeded{%
\ifdim\Gin@nat@width>\linewidth \linewidth \else \Gin@nat@width
\fi } \makeatother

\begin{document}
%

\title{Edge Caching in Dense Heterogeneous Cellular Networks with Massive MIMO Aided Self-backhaul}

\author{Lifeng Wang,~\IEEEmembership{Member,~IEEE,}  Kai-Kit Wong,~\IEEEmembership{Fellow,~IEEE,}  Sangarapillai Lambotharan,~\IEEEmembership{Senior Member,~IEEE,} Arumugam Nallanathan,~\IEEEmembership{Fellow,~IEEE,} and Maged Elkashlan,~\IEEEmembership{Member,~IEEE}
\thanks{This work was supported by the U.K. Engineering and Physical Sciences Research Council (EPSRC) under Grant EP/M016005/1, EP/M016145/2, EP/M015475/1.}
\thanks{L. Wang and K.-K. Wong are with the Department of Electronic and
Electrical Engineering, University College London, London, UK (Email: \{lifeng.wang, kai-kit.wong\}@ucl.ac.uk).}
\thanks{S. Lambotharan is with School of Mechanical, Electrical and Manufacturing Engineering, Loughborough University, Loughborough Leicestershire, UK (Email: $\rm\{s.lambotharan\}@lboro.ac.uk$).}
\thanks{A. Nallanathan and M. Elkashlan are with the School of Electronic Engineering and Computer Science, Queen Mary University of London, London, UK. (Email: $\rm\{a.nallanathan, maged.elkashlan\}@qmul.ac.uk$).}
}

\maketitle

\begin{abstract}
This paper focuses on edge caching in dense heterogeneous cellular networks (HetNets), in which small base stations (SBSs) with limited cache size store the popular contents, and massive multiple-input multiple-output (MIMO) aided macro base stations  provide wireless self-backhaul when SBSs require the non-cached contents. Our aim is to address the effects of cell load and hit probability on the successful content delivery (SCD), and present the minimum required base station density for avoiding the access overload in an arbitrary  small cell and backhaul overload in an arbitrary macrocell. The  achievable rate of massive MIMO backhaul without any downlink channel estimation is derived to calculate the backhaul time, and the latency is also evaluated in such networks. The analytical results confirm that hit probability needs to be appropriately selected, in order to achieve SCD. The interplay between cache size and SCD is explicitly quantified.  It is theoretically  demonstrated that when non-cached contents are requested, the average delay of the non-cached content delivery could be comparable to the cached content delivery with the help of massive MIMO aided self-backhaul, if the average access rate of cached content delivery is lower than that of self-backhauled content delivery. Simulation results are presented to validate our analysis.
\end{abstract}


\section{Introduction}

\subsection{Motivation and Background}
New findings from  Cisco~\cite{CIS-VISUAL} indicate that  mobile video traffic accounts for the majority of mobile data traffic. To offload the traffic of the core networks and reduce the backhaul cost and latency, caching the popular contents at the edge of wireless networks becomes a promising solution~\cite{Paschos_mag_2016,A_liu_2016,lifeng_mag_2017}. The latest 3GPP standard requires that the fifth generation (5G) system shall support content caching applications and operators need to place the content caches close to mobile terminals~\cite{Caching_3GPP_2017}. In addition, the emerging radio-access technologies and wireless network architectures  provide edge caching with new opportunities~\cite{D_liu_survey}.

Recent works have focused on the caching design and analysis in various scenarios. In \cite{Blaszczyszyn_2015}, a probabilistic caching model was considered in single-tier cellular networks and the optimal content placement was designed to maximize the total hit probability. In  \cite{BZHO-STO}, a stochastic content multicast scheduling problem was formulated to jointly minimize the average network delay and power costs in heterogeneous cellular networks (HetNets), and a structure-aware optimal algorithm was proposed to solve this problem. Caching cooperation in multi-tier HetNets was studied in \cite{Xiuhua_2017}, where a low-complexity suboptimal solution was developed to maximize the capacity in such networks. Caching in device-to-device (D2D) networks was investigated in the literature such as \cite{M_Ji_JSAC_2016,zheng_chen_March_2017}. In \cite{M_Ji_JSAC_2016},  a holistic design on D2D caching at multi-frequency band including sub-6 GHz and millimeter wave (mmWave) was presented. In \cite{zheng_chen_March_2017}, the performance difference between maximizing  hit probability and  maximizing cache-aided throughput in D2D caching networks was evaluated.
 The work of  \cite{Zheng_Gan_2017}  showed that in multi-hop relaying systems, the efficiency of caching could be further improved by using collaborative cache-enabled relaying. Joint design of cloud and edge caching in fog radio access networks were introduced in  \cite{S_H_Park_2016,haijunzhang_mag_2017}, where the popular contents were cached at the remote radio heads. However,  prior works~\cite{Blaszczyszyn_2015,BZHO-STO,Xiuhua_2017,M_Ji_JSAC_2016,zheng_chen_March_2017,Zheng_Gan_2017,S_H_Park_2016} did not present design and insights involving edge caching in the future dense/ultra-dense cellular networks (e.g., 5G) with backhaul limitations, where wireless self-backhauling shall be supported~\cite{lifeng_mag_2017}.

  Cache-enabled small cell networks with stochastic models have been investigated in the literature such as \cite{ZCH-COO,Youjia_chen_2017,KUI-OPT,J_Wen_2017,shijin_caching_2017}. Cluster-centric caching with base station (BS) cooperation was studied in \cite{ZCH-COO}, where the tradeoff between transmission diversity and content diversity was revealed. In \cite{Youjia_chen_2017} where it was assumed that the intensity of BSs is much larger than the intensity of mobile terminals, two cache-enabled BS modes were considered, namely always-on and dynamic on-off. The work of \cite{KUI-OPT,J_Wen_2017,shijin_caching_2017} concentrated on the cache-enabled multi-tier HetNets. Specifically, \cite{KUI-OPT} and \cite{J_Wen_2017} studied optimal content placement under probabilistic caching strategy, and  \cite{shijin_caching_2017} considered the joint BS caching and cooperation, in contrast to the single-tier case in \cite{ZCH-COO}. However,   \cite{ZCH-COO,Youjia_chen_2017,KUI-OPT,J_Wen_2017,shijin_caching_2017} only aimed to maximize the probability that the requested content is not only cached but also successfully delivered. In realistic networks, when the requested contents of users are not cached at their associated BSs, they will be obtained from the core network via wired/wireless backhaul, which also needs to be studied in cache-enabled cellular networks.

  In fact, existing contributions such as \cite{MTAO-CON,X_Peng,Anliu_IEEE_ACM} have studied the effects of backhaul on content delivery in cache-enabled networks. The work of \cite{MTAO-CON} considered that non-cached contents were obtained via backhaul, and  a downlink content-centric sparse multicast beamforming was proposed for the cache-enabled cloud radio access network (Cloud-RAN), to minimize the weighted sum of backhaul cost and transmit power. In \cite{X_Peng},  the network successful content delivery consisting of cached content delivery and backhauled content delivery was studied. The optimization problem was formulated to minimize the cache size under quality-of-service constraint. The work of \cite{Anliu_IEEE_ACM} analyzed the capacity scaling law when there are limited number of wired backhaul in single-tier networks, and showed that cache size needs to be large enough to achieve linear capacity scaling. However, none of \cite{MTAO-CON,X_Peng,Anliu_IEEE_ACM} has studied the cache-enabled cellular networks with specified wireless backhaul transmission, such as massive multiple-input multiple-output (MIMO) aided self-backhaul.

\subsection{Novelty and Contributions}
In this paper, we focus on the edge caching in dense HetNets with massive MIMO aided self-backhaul, which has not been understood yet. Massive MIMO aided self-backhaul is motivated by the facts that it may not be feasible to have optical fiber for every backhaul channel and massive MIMO can support high-speed transmissions {thanks to large array gains and  multiplexing gains}~\cite{lifeng_mag_2017}. Our contributions are summarized as follows:
\begin{itemize}
  \item  In contrast to the prior works such as \cite{ZCH-COO,Youjia_chen_2017,KUI-OPT,J_Wen_2017,shijin_caching_2017,MTAO-CON,X_Peng,Anliu_IEEE_ACM}, we consider cache-enabled HetNets, in which randomly located small BSs (SBSs) cache finite popular contents, and macro BSs (MBSs) equipped with massive MIMO antennas provide wireless backhaul to deliver the non-cached requested contents to the SBSs. Moreover, we also consider the resource allocation  when multiple users request the contents from the same SBS, which has not been studied in a cache-enabled stochastic model.
  \item We first derive the successful content delivery probability  when the requested content is cached at the SBS. The maximum small cell load is calculated, and the minimum required density of SBSs for avoiding access overload is obtained. We show that hit probability needs to be lower than a critical value, to guarantee successful cached content delivery.

  \item We derive the successful content delivery probability when the requested content is not cached and has to be obtained via massive MIMO backhaul. We analyze the massive MIMO backhaul achievable rate when downlink channel estimation is not {necessary}, to evaluate the backhaul transmission delay. The minimum required density of MBSs for avoiding backhaul overload is obtained. We show that hit probability needs to be higher than a critical value, to guarantee  successful self-backhauled content delivery.

  \item We analyze the effects of cache size on the successful content delivery, and provide important insights on the interplay between time-frequency resource allocation and cache size from the perspective of successful content delivery probability. We characterize the latency in terms of average delay in such networks. We confirm that when the requested contents are not cached, the average delay of the non-cached content delivery could be comparable to the cached content delivery with the assistance of massive MIMO backhaul, if the average access rate of cached content delivery is lower than that of  self-backhauled content delivery.
\end{itemize}

\section{Network Model}
\begin{figure}[t!]
\centering
\includegraphics[width=3.2 in,height=2.8 in]{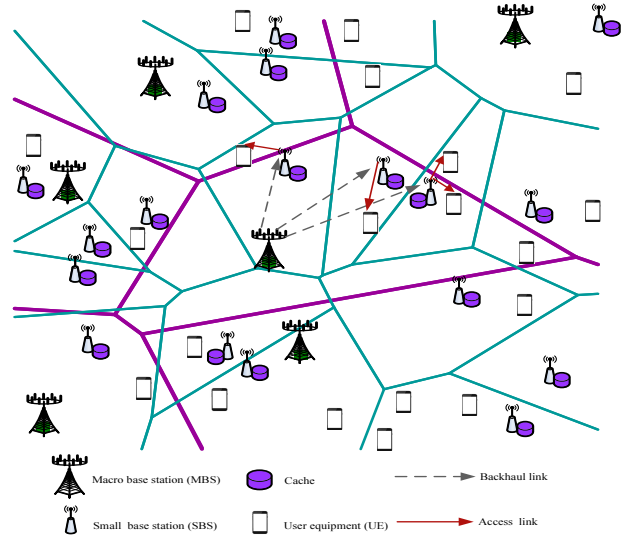}
\caption{An illustration of cache-enabled heterogeneous cellular network with massive MIMO backhaul.}
\label{WPT_UE}
\end{figure}
As shown in Fig. \ref{WPT_UE}, we consider a two-tier self-backhauled HetNet, in which each single-antenna SBS with finite cache size can store popular contents to serve user equipment (UEs). Each massive MIMO aided MBS equipped with $N$ antennas has access to the core network via optical fiber and delivers the non-cached contents to the SBSs via wireless backhaul. UEs, SBSs, and MBSs are assumed to be distributed following independent homogeneous Poisson point processes (HPPPs) denoted by $\Phi_\mathrm{U}$ with the density $\lambda_\mathrm{U}$, $\Phi_\mathrm{S}$ with the density $\lambda_\mathrm{S}$, and $\Phi_\mathrm{M}$ with the density $\lambda_\mathrm{M}$, respectively. It is assumed that UEs are associated with the SBSs that can provide the maximum average received power, which is also utilized in 4G networks~\cite{D_liu_survey}. In addition, each channel
undergoes independent and identically distributed (i.i.d.) quasi-static Rayleigh fading.

\subsection{Content Placement}
Content placement mechanism is mainly designed based on content
popularity~\cite{lifeng_mag_2017}. We assume that there is a finite content library denoted as $\mathcal{F}:=\lbrace f_1,\dots,f_j,\dots,f_J \rbrace$, where $f_j$ is the $j$-th most popular content and the number of contents is $J$. The request probability for the $j$-th most popular content is commonly-modeled by following the {Zipf} distribution~\cite{Breslau_Zif_1999}
\begin{align}
{a}_j = {j^{-\varsigma} }/\sum\nolimits_{m= 1}^J {{m^{-\varsigma} }} ,
\end{align}
where $\varsigma $ is the Zipf exponent to represent the popularity skewness~\cite{Breslau_Zif_1999}. Each content is assumed to be unit size and each SBS can only cache $L \left(L \ll J\right)$ contents. We employ the probabilistic caching strategy~\cite{Blaszczyszyn_2015}, i.e., {the probability that the content $j$ is cached at an arbitrary SBS is $q_j$$(0\leq q_j \leq 1)$, and the sum of probabilities for all the contents being cached at an arbitrary SBS should be less than the SBS's cache size (namely $\sum\limits_{j{\rm{ = }}1}^J {{q_j}} {\rm{ \leq }}L$)~\cite{Blaszczyszyn_2015}.} { Note that based on the probabilistic caching strategy, each SBS only stores $L$ files from the content library for each caching realization, which is also illustrated in Fig.1 of~\cite{Blaszczyszyn_2015}. }


\subsection{Self-backhaul Load}
We assume that the access and backhaul links share the same sub-6 GHz spectrum. The bandwidths allocated to the access and backhaul links are $\eta W$ and $\left(1-\eta\right) W$, respectively, where $\eta$ is the fraction factor and $W$ is the system bandwidth. The number of UEs that is associated with an SBS is denoted by $K$.  The UEs in the same small cell are  served in a time-division manner with equal-time sharing. Thus, the fraction of time-frequency resources allocated to each access link is $\eta W/K$ during the cached content delivery. {When an associated SBS does not cache the requested content, it has to be connected to an MBS that provides the strongest wireless backhaul link such that the requested content can be obtained from core networks. Note that different SBSs may be served by different MBSs. Let $S_j \left(N \gg S_j\right)$ denote the number of SBSs served by the $j$-th MBS $\left(j\in \Phi_\mathrm{M}\right)$ for wireless backhaul.}

{ Hit probability characterizes the probability that a requested content file is stored at an arbitrary SBS~\cite{lifeng_mag_2017}, and is calculated as $q_{\mathrm{hit}}=\sum\limits_{j{\rm{ = }}1}^J {{a_j}{q_j}}$.}  The set of SBSs can be partitioned into two independent HPPPs $\Phi_\mathrm{S}^\mathrm{a}$ and $\Phi_\mathrm{S}^\mathrm{b}$ based on the thinning theorem~\cite{Baccelli2009}, where $\Phi_\mathrm{S}^\mathrm{a}$ with the density $\lambda_\mathrm{S} q_{\mathrm{hit}}$ denotes the point process of SBSs with access links, and $\Phi_\mathrm{S}^\mathrm{b}$  with the density $\lambda_\mathrm{S}\left(1-q_{\mathrm{hit}}\right)$ denotes the point process of SBSs with backhaul links. Let $\omega_\mathrm{b}=\lambda_\mathrm{S}\left(1-q_{\mathrm{hit}}\right)/\lambda_\mathrm{M}$ represent the average number of SBSs served by an MBS for wireless backhaul.

\subsection{Resource Allocation Model}
We consider the saturated traffic condition, i.e.,  all the SBSs keep active to serve their associated UEs.
\subsubsection{Access} When the requested content is stored at a typical SBS, the  rate for a typical access link is given by
\begin{align}\label{access_Rate}
R_{\rm{a}}=\frac{\eta W}{K}\log_2\Bigg(1+\frac{{{P_\mathrm{a}}{h_o}L\left( {\left| {{X_o}} \right|} \right)}}{{ \underbrace{\sum\limits_{i \in {\Phi_\mathrm{S}^\mathrm{a}}\backslash \left\{ o \right\}} {{P_\mathrm{a}}{h_i}L\left( {\left| {{X_{o{\rm{,}}i}}} \right|} \right)}}_{I_\mathrm{a}} + \sigma _\mathrm{a}^2}}\Bigg),
\end{align}
{where $I_\mathrm{a}$ denotes the total interference power from other SBSs; $P_\mathrm{a}$ is the SBS's transmit power; $L\left(\left|X\right|\right)=\beta \left(\left|X\right|\right)^{-\alpha_\mathrm{a}}$ denotes the path loss with frequency dependent constant value $\beta$, distance $\left|X\right|$ and path loss exponent $\alpha_\mathrm{a}$;   $h_o\sim \rm{exp}(1)$ and $\left| {{X_o}} \right|$ are the small-scale fading channel power gain and distance between the typical UE and its associated SBS, respectively; $h_i\sim \rm{exp}(1)$ and $\left| {{X_{o{\rm{,}}i}}} \right|$ are the small-scale fading interfering channel power gain and distance between the typical UE and the interfering SBS $i\in {\Phi_\mathrm{S}^\mathrm{a}}\backslash \left\{ o \right\}$ (except the typical SBS $o$) respectively, and $\sigma _a^2$ is the noise power at the typical UE.}

\subsubsection{Self-Backhaul} When the requested content is not stored at SBSs, it is obtained through massive MIMO backhaul.
For massive MIMO backhaul link, we consider that massive MIMO enabled MBS adopts zero-forcing beamforming with equal power allocation~\cite{Hosseini2014_Massive}. {In such a time-division duplex (TDD) massive MIMO self-backhauled network, SBSs will not perform any channel estimation\footnote{In TDD massive MIMO systems, downlink precoder is designed based on the uplink channel estimation, thanks to channel reciprocity~\cite{Bjorson_mag_2016}. Moreover, when deploying massive number of antennas at MBS, wireless channel behaves in a deterministic manner called ``channel hardening'', and the effect of small-scale fading could be negligible~\cite{hien_2018_pilot}. Thus, in the considered massive MIMO backhaul scenario where MBSs and SBSs are usually still, the coherence time of backhaul channel will be much longer than ever before, and the time occupied by uplink channel estimation will be much lower.

{It should be noted that when high-mobility UEs  are served by TDD massive MIMO BSs, downlink pilots may still be needed to estimate the fast-changing channels~\cite{Malkowsky_2017_access}.}}, and we will adopt an achievable backhaul transmission rate as confirmed in~\cite{J_JOSE_2011,J_Hoydis_2013}.}   Based on the instantaneous received signal expression in~\cite[Eq. 6]{J_JOSE_2011}, given a typical distance $\left| {{Y_o}} \right|$ between the typical SBS and its associated MBS,  the instantaneous rate for a typical massive MIMO backhaul link is given by
\begin{align}\label{rate_backhaul}
R_{\rm{b}}={\left(1-\eta\right) W}\log_2\left(1+\mathrm{SINR_b}\right)
\end{align}  
with
\begin{align*}
&\hspace{-0.2 cm}\mathrm{SINR_b}= \nonumber\\
&\hspace{-0.3 cm}\frac{{\frac{{{P_\mathrm{b}}}}{S_o}{\left(\mathbb{E}\left\{\sqrt{{g_o}  }\right\}\right)^2}L\left( {\left| {{Y_o}} \right|} \right)}}
{{\frac{{{P_\mathrm{b}}}}{S_o}\left(\sqrt{{g_o}  }-\mathbb{E}\left\{\sqrt{{g_o}  }\right\}\right)^2  L\left( {\left| {{Y_o}} \right|} \right)+\underbrace{\sum\limits_{j \in \Phi _{\rm{M}}\backslash \left\{ o \right\}} {\frac{{{P_\mathrm{b}}}}{S_j}{g_j}L\left( {\left| {{Y_{o{\rm{,}}j}}} \right|} \right)}}_{I_\mathrm{b}}+\sigma_\mathrm{b}^2}},
\end{align*}
{where $\mathbb{E}\left\{\cdot\right\}$ is the expectation operator. Here,  $I_\mathrm{b}$ denotes the total interference power from other MBSs; $P_\mathrm{b}$ is the MBS's transmit power; $L\left(\left|Y \right|\right)=\beta \left(\left|Y \right|\right)^{-\alpha_\mathrm{b}}$ denotes the path loss with the distance $\left|Y\right|$ and path loss exponent $\alpha_\mathrm{b}$; $g_o \sim \Gamma\left(N-S_o+1,1\right)$ is the small-scale fading channel power gain between the typical SBS and its associated MBS;  $g_j\sim \Gamma\left(S_j,1\right)$\footnote{$\Gamma\left(\cdot,\cdot\right)$ is the upper incomplete gamma function~\cite[(8.350)]{gradshteyn}.}, and $\left| {{Y_{o{\rm{,}}j}}} \right|$ are the small-scale fading interfering channel power gain and distance between the typical SBS and interfering MBS $j$, respectively, and $\sigma _\mathrm{b}^2$ is the noise power at the typical SBS.}

After obtaining the requested content via backhaul, the associated SBS delivers it to the corresponding UE. In this case, the corresponding access-link rate is expressed as
\begin{align}\label{AL_backhaul}
R_{\mathrm{a}^{'}}&=\frac{\left(1-\eta\right) W}{K} \times  \nonumber\\
&\qquad\log_2\Bigg(1+\frac{{{P_{\mathrm{a}}}{h_o}L\left( {\left| {{X_o}} \right|} \right)}}{{ \underbrace{\sum\limits_{{i^{'}} \in {\Phi_\mathrm{S}^\mathrm{b}}\backslash \left\{ o \right\}} {{P_{\mathrm{a}}}{h_{i^{'}}}L\left( {\left| {{X_{o{\rm{,}}i^{'}}}} \right|}
 \right)}}_{I_{\mathrm{a}^{'}}} + \sigma _{\mathrm{a}^{'}}^2}}\Bigg),
\end{align} 
where $I_{\mathrm{a}^{'}}$ is the total interference power,  $h_{i^{'}}\sim \rm{exp}(1)$ and $L\left( {\left| {{X_{o{\rm{,}}i^{'}}}} \right|} \right)=\beta \left(\left|{{X_{o{\rm{,}}i^{'}}}} \right|\right)^{-\alpha_\mathrm{a}}$ are the small-scale fading channel power gain and pathloss between the typical SBS and interfering SBS ${i^{'}} \in {\Phi_\mathrm{S}^\mathrm{b}}\backslash \left\{ o \right\}$, respectively, and $\sigma _{\mathrm{a}^{'}}^2$ is the noise power at the typical UE.

From \eqref{rate_backhaul} and \eqref{AL_backhaul}, {we see that to reduce latency, massive MIMO backhaul link needs to be of   high-speed, which can be achieved by using large antenna arrays at the MBS.}  In the following section, we will further examine how much backhaul time is needed at an achievable backhaul rate.

\section{Content Delivery Efficiency}
{In this paper, there are two cases for successful content delivery (SCD), i.e., 1) when the associated BS has cached the requested content,  SCD occurs if the time for successfully delivering $Q$ bits will not exceed the threshold $T_\mathrm{th}$; and 2) when the requested content is not cached at the associated BS and needs to be obtained via massive MIMO backhaul, SCD occurs if the total time for successfully delivering $Q$ bits to the UE is less than $T_\mathrm{th}$.}
\subsection{Cached Content Delivery}
Different from \cite{ZCH-COO,Youjia_chen_2017,J_Wen_2017} where it is assumed that each small cell has only one active UE, we evaluate SCD probability by considering multiple UEs served by an SBS, and analyze the effect of resource allocation on SCD probability. We first have the following important theorem.
\begin{theorem}
When a requested content is stored at the typical SBS, the SCD probability is derived as
\begin{align}\label{SCDP_cached}
&{\Psi_{\mathrm{SCD}}^\mathrm{a}}\left( {Q,T_\mathrm{th}} \right)=\sum\limits_{k=1}^{{K_{\max}^{\rm{a}}}} \mathcal{P}_{{\frac{\lambda_\mathrm{U}}{\lambda_\mathrm{S}}}}\left(k\right),
\end{align}
where $\mathcal{P}_{\frac{\lambda_\mathrm{U}}{\lambda_\mathrm{S}}}\left(k\right)$ is the probability mass function (PMF) that there are other $k-1$ UEs (except typical UE) served by the typical SBS, and is given by $\mathcal{P}_{{\frac{\lambda_\mathrm{U}}{\lambda_\mathrm{S}}}}\left(k\right)=\frac{\gamma^\gamma}{\left(k-1\right)!}
\frac{\Gamma\left(k+\gamma\right)}{\Gamma\left(\gamma\right)} \frac{\left(\frac{\lambda_\mathrm{U}}{\lambda_\mathrm{S} }\right)^{k-1}}{ \left(\gamma+\frac{\lambda_\mathrm{U}}{\lambda_\mathrm{S}}\right)^{k+\gamma}   }$ with $\gamma=3.5$~\cite{J_S_Ferenc_2007}. In \eqref{SCDP_cached}, $K=K_{\max }^{\rm{a}}$ is the maximum load in a typical small cell, and can be efficiently obtained by using \textbf{Algorithm 1} to solve the following equation
\begin{align}\label{K_max_eq_11}
\frac{2^{\frac{{K_{\max }^{\rm{a}} Q}}{\eta W T_\mathrm{th}}+1}-2}{\alpha_\mathrm{a}-2}\chi_k^{\mathrm{a}}\left(K_{\max }^{\rm{a}}\right) =\frac{1-\epsilon}{q_{\mathrm{hit}}\epsilon},
\end{align}
where $\chi_k^{\mathrm{a}}\left(K_{\max }^{\rm{a}}\right)={_2F_{1}}\left[1,1-\frac{2}{\alpha_\mathrm{a}};2-\frac{2}{\alpha_\mathrm{a}};
1-{2^{\frac{{K_{\max }^{\rm{a}} Q}}{{\eta W T_\mathrm{th}  }}}} \right]$, {$_2{F_1}\left[\cdot,\cdot;\cdot;\cdot\right]$
is the Gauss hypergeometric function~\cite[(9.142)]{gradshteyn}\footnote{ In MATLAB R2015b software, hypergeom([a,b],c,z) is the Gauss hypergeometric function $_2{F_1}\left[a,b;c;z\right]$.}}, and $\epsilon$ is the predefined threshold, i.e., SCD occurs when the probability that $R_\mathrm{a}$ is larger than $\frac{Q}{T_\mathrm{th}}$ is above  $\epsilon$.
\begin{table}[htbp]
\centering
\begin{tabular}{l}
\hline
{\textbf{Algorithm 1} One-dimension Search }\\ \hline
1: $\bf{if}$ $t=0$   \\
2: \quad Initialize $\varphi=\frac{1-\epsilon}{q_{\mathrm{hit}}\epsilon}$, $k^{l}=1, k^{h}=10 \times\frac{\lambda_\mathrm{U}}{\lambda_\mathrm{S}}$, and calculate \\
~~~~$F^{l}=\frac{2^{\frac{{k^{l} Q}}{\eta W T_\mathrm{th}}+1}-2}{\alpha_\mathrm{a}-2} {_2F_{1}}\left[1,1-\frac{2}{\alpha_\mathrm{a}};2-\frac{2}{\alpha_\mathrm{a}};
1-{2^{\frac{{k^{l} Q}}{{\eta W T_\mathrm{th}  }}}} \right]$ \\
and \\
~~~~$F^{h}=\frac{2^{\frac{{k^{h} Q}}{\eta W T_\mathrm{th}}+1}-2}{\alpha_\mathrm{a}-2} {_2F_{1}}\left[1,1-\frac{2}{\alpha_\mathrm{a}};2-\frac{2}{\alpha_\mathrm{a}};
1-{2^{\frac{{k^{h} Q}}{{\eta W T_\mathrm{th}  }}}} \right]$  \\
3: \bf{else} \\
4: \quad  \textbf{While} $F^{l} \neq \varphi$ and $F^{h} \neq \varphi$ \\
5: \qquad  Let $k=\frac{k^{l}+k^{h}}{2}$, and compute $F_k$. \\
6: \qquad $\bf{if}$ $F_k=\varphi$ \\
7: \qquad~~~ The optimal $k^{*}$ is obtained, i.e., $K_{\max}^{\rm{a}}=round\left( {k^{*}}\right)$. \\
8: \qquad~~~ \bf{break}   \\
9: \qquad \bf{elseif} $F_k<\varphi$ \\
10: \qquad~~~ $k^{l}=k$. \\
11: \qquad \bf{else} $F_k>\varphi$ \\
12: \qquad~~~$k^{h}=k$. \\
13: \qquad \bf{end if}  \\
14: \quad \bf{end while} \\
15: \bf{end if}  \\
\hline
\end{tabular}
\label{table:1}
\end{table}
\begin{proof}
See Appendix A.
\end{proof}

\end{theorem}

It is implied from \textbf{Theorem 1} that in the dense small cell networks (i.e., interference-limited)\footnote{The near-field pathloss exponent is assumed to be larger than 2~\cite{lifeng_mag_2017}.}, the SCD probability depends on the ratio of UE density to SBS density and hit probability given the time-frequency resource allocation. Based on \textbf{Theorem 1}, we have
\begin{corollary}
From \eqref{K_max_eq_11}, we see that to achieve the load $K=K_{\max }^{\rm{a}} \geq 1$ in a small cell, the hit probability should satisfy
\begin{align}\label{hit_pro}
q_{\mathrm{hit}} \leq \min\left\{ \Xi_{\rm{a}} \frac{1-\epsilon}{\epsilon},1\right\},
\end{align}
where $\Xi_{\rm{a}}=\left(\frac{2^{\frac{{ Q}}{\eta W T_\mathrm{th}}+1}-2}{\alpha_\mathrm{a}-2}\chi_k^{\mathrm{a}} \left(1\right)\right)^{-1}$.

It is indicated from \eqref{hit_pro} that there is an upper-bound on the hit probability, which can be explained by the fact that when more UEs can obtain their requested contents from their associated SBSs in dense cellular networks with large hit probability, there will also be more interference from nearby SBSs  that hinders the cached content delivery.
\end{corollary}


In realistic networks, there may be overload issues when the scale of small cells is not adequate to support large level of connectivity, which needs to be addressed. Therefore, given a specified scale of UEs $\lambda_\mathrm{U}$, we evaluate the minimum required scale of small cells as follows.
 \begin{corollary}
To mitigate the harm of overloading,  the minimum required SBS density needs to satisfy
\begin{align}\label{SBS_density_min_theo1}
{\lambda_\mathrm{S}}=\left\{ \begin{array}{l}
\frac{\lambda_\mathrm{U}}{K_{\max}^{\rm{a}}+1},\; if \;{\mathcal{P}_{\frac{\lambda_\mathrm{U}}{\lambda_\mathrm{S}}=K_{\max}^{\rm{a}}+1}\left(K_{\max}^{\rm{a}}+1\right)} \leq \rho,
\\
\frac{\lambda_\mathrm{U}}{\mu_{\rm{a}}},\; if \;{\mathcal{P}_{\frac{\lambda_\mathrm{U}}{\lambda_\mathrm{S}}=K_{\max}^{\rm{a}}+1}\left(K_{\max}^{\rm{a}}+1\right)} > \rho,
\end{array} \right.
\end{align}
where $\mu_{\rm{a}}\in \Big(0,\frac{K_{\max}^{\rm{a}}\gamma}{\gamma+1}\Big]$ is the solution of $\mathcal{P}_{\frac{\lambda_\mathrm{U}}{\lambda_\mathrm{S}}=\mu_{\rm{a}}}\left(k=K_{\max}^{\rm{a}}+1\right)=\rho$ with arbitrary small $\rho>0$, and can be easily obtained via one-dimension search, similar to \textbf{Algorithm 1}. Such network deployment given in \eqref{SBS_density_min_theo1} can guarantee
$\mathcal{P}_{\frac{\lambda_\mathrm{U}}{\lambda_\mathrm{S}}}\left(k\right)\leq\rho, \forall k>K_{\max}^{\rm{a}}$.
\end{corollary}
\begin{proof}
See Appendix B.
\end{proof}

From \eqref{SBS_density_min_theo1}, we see that the minimum required density of SBSs only depends on   the maximum load of a small cell and the density of UEs in dense cache-enabled cellular networks.

\subsection{Self-backhauled Content Delivery}
\subsubsection{Massive MIMO Backhaul}When the required content is not stored at the typical SBS, SBS has to obtain it from the core network via massive MIMO backhaul. Therefore, we need to evaluate the backhaul time for delivering the requested content to the typical SBS.  {It should be noted that the load in a macrocell will not change fast,  in order to deliver the requested contents to the associated SBSs. Hence, given the load $S_o$ in a typical macrocell, the achievable transmission rate for a typical backhaul link is given by}
\begin{align}\label{backhaul_tran_rate_exp}
&\overline{R}_\mathrm{b}\left(S_o\right)={\left(1-\eta\right) W} \int_{r_\mathrm{b}}^\infty C_\mathrm{b}\left(y\right)\frac{2 \pi \lambda_\mathrm{M} y e^{-\pi \lambda_\mathrm{M} y^2}}{e^{-\pi\lambda_\mathrm{M} r_\mathrm{b}^2}}  dy,
\end{align}
where  $C_\mathrm{b}\left(y\right)=\log_2\Bigg(1+\frac{\frac{{{P_\mathrm{b}}}}{S_o} \Xi_1\left(y\right)}{{\frac{{{P_\mathrm{b}}}}{S_o}\Xi_2\left(y\right)+\Xi_3\left(y\right)+\sigma_\mathrm{b}^2}}\Bigg)$ with $\Xi_1\left(y\right)=L\left( y \right) \left(\frac{ \Gamma \left( {N - S_o + \frac{3}{2}} \right)}{\Gamma \left( {N - S_o + 1} \right)}\right)^2$, $\Xi_2\left(y\right)=(N-S_o+1)L\left( y \right)-\Xi_1$, and $\Xi_3\left(y\right)=P_\mathrm{b} 2 \pi \lambda_\mathrm{M} \beta \frac{y^{2-\alpha_\mathrm{b}} }{\alpha_\mathrm{b}-2}$, and $r_\mathrm{b}$ is the minimum distance between the typical MBS and its associated SBS. A detailed derivation of \eqref{backhaul_tran_rate_exp} is provided in Appendix C. Therefore, the time for delivering $Q$ bits to the typical SBS via wireless backhaul is $T_1=\frac{Q}{\overline{R}_\mathrm{b}}$.  When the number of antennas at the  MBS grows large, we have the following corollary.
\begin{corollary}\label{time_constraint}
 For large $N$, the achievable transmission rate for a typical backhaul link is tightly lower-bounded as
 {\begin{align}\label{coro_3_eq_lowerbound}
 \overline{R}_\mathrm{b}^{\mathrm{Low}}\left(S_o\right)={\left(1-\eta\right) W} \log_2\left(1+P_\mathrm{b}\beta\frac{N-S_o+\frac{1}{2}}{S_o}e^{\overline{\Delta}_1-\overline{\Delta}_2} \right),
 \end{align}}
where
 \begin{equation}\label{Delta_1_2_asympt}
 \left\{\begin{aligned}
 & \overline{\Delta}_1=-{\alpha_\mathrm{b}}{e^{\pi\lambda_\mathrm{M} r_\mathrm{b}^2}}\Big(-\frac{Ei\big(-r_\mathrm{b}^2 \pi \lambda_\mathrm{M} \big)}{2}+e^{-r_\mathrm{b}^2 \pi \lambda_\mathrm{M}}\ln r_\mathrm{b} \Big), \\
 &\overline{\Delta}_2=\int_{r_\mathrm{b}}^\infty \ln\left(\frac{P_\mathrm{b}\beta}{2S_o}y^{-r_\mathrm{b}}+P_\mathrm{b} 2 \pi \lambda_\mathrm{M} \beta \frac{y^{2-\alpha_\mathrm{b}} }{\alpha_\mathrm{b}-2}+\sigma_\mathrm{b}^2\right) \nonumber\\
&\;\;\;\qquad\;\; \times \frac{2 \pi \lambda_\mathrm{M} y}{e^{-\pi\lambda_\mathrm{M} r_\mathrm{b}^2}} e^{-\pi \lambda_\mathrm{M} y^2} dy,
\end{aligned}\right.
 \end{equation}
 in which $Ei\left(z\right)$ is the exponential integral given by $Ei\left(z\right)=-\int_{-z}^\infty \frac{e^{-t}}{t} dt$~\cite{gradshteyn}. Based on \eqref{coro_3_eq_lowerbound}, the typical MBS's required time for  delivering $Q$ bits to its associated SBS  satisfies
 \begin{align}\label{backhaul_time_constraint}
 T_1 \leq \frac{{Q}{\left(1-\eta\right)^{-1} W^{-1}}} {\log_2\left(1+P_\mathrm{b}\beta\frac{\left(N-S_o+\frac{1}{2}\right)}{S_o}e^{\overline{\Delta}_1-\overline{\Delta}_2} \right)}.
\end{align}
 \end{corollary}
\begin{proof}
See Appendix D.
\end{proof}

It is explicitly shown from \textbf{Corollary 3} that large number of antennas and bandwidths are required, in order to significantly reduce the wireless backhaul delivery time. From \eqref{backhaul_time_constraint}, we see that the backhaul delivery time can at least be cut proportionally to { $1/\log_2 N$}.

In the self-backhauled networks, the number of SBSs  being simultaneously served by an MBS for wireless backhaul should not exceed the maximum value denoted by $S_{\max}$, i.e., $S_o \leq S_{\max}$; otherwise high-speed massive MIMO aided backhaul transmission cannot be guaranteed. Hence, given the minimum  required backhaul transmission rate $R_\mathrm{b}^{\min}$, the maximum backhaul load of a typical massive MIMO MBS  is the solution of $\overline{R}_\mathrm{b}\left(S_{\max}\right)=R_\mathrm{b}^{\min}$, which can be efficiently obtained by using one-dimension search since $\overline{R}_\mathrm{b}\left(S_o\right)$ is a decreasing function of $S_o$ for large $N$, as suggested in Appendix D. After obtaining $S_{\max}$, we can obtain the minimum number of massive MIMO aided MBSs that needs to be deployed, in order to mitigate the backhaul overload.
\begin{corollary}
 Similar to \textbf{Corollary 2}, the minimum required density of MBSs is given by
\begin{align}\label{backhaul_massive_MIMO_density}
\lambda_\mathrm{M}=\left\{ \begin{array}{l}
\frac{\lambda_\mathrm{S}\left(1-q_{\rm{hit}}\right)}{S_{\max}+1},\; if \;\mathcal{P}_{\omega_\mathrm{b}=S_{\max}+1}\left(S_{\max}+1\right) \leq \rho,
\\
\frac{\lambda_\mathrm{S}\left(1-q_{\rm{hit}}\right)}{\mu_{\rm{b}}},\; if \;\mathcal{P}_{\omega_\mathrm{b}=S_{\max}+1}\left(S_{\max}+1\right) > \rho,
\end{array} \right.
\end{align}
where $\mathcal{P}_{\omega_\mathrm{b}}\left(\ell\right)=\frac{\gamma^\gamma}{\left(\ell-1\right)!}
\frac{\Gamma\left(\ell+\gamma\right)}{\Gamma\left(\gamma\right)} \frac{\left(\omega_\mathrm{b}\right)^{\ell-1}}{ \left(\gamma+\omega_\mathrm{b}\right)^{\ell+\gamma}   }$, $\mu_{\rm{b}}\in \Big(0,\frac{S_{\max}\gamma}{\gamma+1}\Big]$ is the solution of $\mathcal{P}_{\omega_\mathrm{b}=\mu_{\rm{b}}}\left(S_{\max}+1\right)=\rho$ with arbitrary small $\rho>0$, and can be easily obtained via one-dimension search.
\end{corollary}

It is explicitly shown in \eqref{backhaul_massive_MIMO_density} that higher hit probability can significantly reduce the scale of MBSs because of less backhaul.

\subsubsection{Access} After obtaining the required content via backhaul, the typical SBS transmits it to the associated UE. Thus, we have the following important theorem.
\begin{theorem}
When the required content is not stored at the typical SBS and has to be obtained via massive MIMO self-backhaul, the SCD probability is derived as
\begin{align}\label{SCDP_backhaul_pro}
{\Psi_{\mathrm{SCD}}^\mathrm{b}}\left( {Q,T_\mathrm{th}} \right)=\sum\limits_{k=1}^{{K_{\max}^{\rm{b}}}} \mathcal{P}_{{\frac{\lambda_\mathrm{U}}{\lambda_\mathrm{S}}}}\left(k\right),
\end{align}
where $K_{\max}^{\rm{b}}$ is the maximum number of UEs that a typical small cell can serve when the typical UE's content needs to be attained via backhaul, and $K_{\max}^{\rm{b}}$ can be obtained by solving the following equation{\footnote{It can be  solved by following \textbf{Algorithm 1}.}}
\begin{align}\label{K_b_num}
\frac{2^{\frac{{K_{\max}^{\rm{b}}Q}}{{\left(1-\eta\right) W \left(T_\mathrm{th}-T_1\right)  }}+1}-2}{\alpha_\mathrm{a}-2}\chi_k^{\mathrm{b}}\left(K_{\max}^{\rm{b}}\right)=\frac{1-\epsilon}{\left(1-q_{\mathrm{hit}}\right)\epsilon}
\end{align}
with $\chi_k^{\mathrm{b}}\left(K_{\max}^{\rm{b}}\right)={_2F_{1}}\left[1,\frac{\alpha_\mathrm{a}-2}{\alpha_\mathrm{a}};\frac{2\alpha_\mathrm{a}-2}{\alpha_\mathrm{a}};
1-{2^{\frac{{K_{\max}^{\rm{b}}Q}}{{\left(1-\eta\right) W \left(T_\mathrm{th}-T_1\right)  }}}} \right]$, and the minimum required SBS density for mitigating overload is given from \eqref{SBS_density_min_theo1} by interchanging $K_{\max}^{\rm{a}}\rightarrow K_{\max}^{\rm{b}}$.
\begin{proof}
See Appendix E.
\end{proof}

\end{theorem}

It is indicated from \eqref{K_b_num} that when a typical UE's requested content is not stored at the typical SBS, the number of UEs that can be served by the typical SBS decreases with increasing backhaul time.
Based on \textbf{Theorem 2}, we have the following corollary
\begin{corollary}
From \eqref{K_b_num}, we see that to achieve the load $K=K_{\max }^{\rm{b}} \geq 1$ in a small cell, the hit probability should satisfy
\begin{align}\label{hit_pro_back}
q_{\mathrm{hit}} \geq \left[1- \Xi_{\rm{b}}\frac{1-\epsilon}{\epsilon}\right]^{+},
\end{align}
where $\Xi_{\rm{b}}=\left(\frac{2^{\frac{{ Q}}{\left(1-\eta\right) W \left(T_\mathrm{th}-T_1\right)}+1}-2}{\alpha_\mathrm{a}-2}\chi_k^{\mathrm{b}} \left(1\right)\right)^{-1}$, and $\left[x\right]^{+}=\max\left\{x,0\right\}$.
\end{corollary}

From \eqref{hit_pro_back}, we see that there is a lower-bound on the hit probability, i.e., minimum cache capacity is demanded at the SBS, since more backhaul results in more interference, which  will degrade the self-backhauled content delivery.

\begin{corollary}
After obtaining the maximum load $K_{\max}^{\rm{b}}$, we can calculate the minimum required SBS density given from \eqref{SBS_density_min_theo1} by interchanging $K_{\max}^{\rm{a}}\rightarrow K_{\max}^{\rm{b}}$, to overcome overload.
\end{corollary}

Based on \textbf{Theorem 1} and \textbf{Theorem 2}, the SCD probability in dense cellular networks with massive MIMO self-backhaul for a typical UE is calculated as
\begin{align}\label{SCD_massiveMIMO_backhaul_access}
&\hspace{-0.2cm} {\Psi_{\mathrm{SCD}}}\left( {Q,T_\mathrm{th}} \right)=q_\mathrm{hit}{\Psi_{\mathrm{SCD}}^\mathrm{a}}\left( {Q,T_\mathrm{th}} \right)+\left(1-q_\mathrm{hit}\right){\Psi_{\mathrm{SCD}}^\mathrm{b}}\left( {Q,T_\mathrm{th}} \right)\nonumber\\
 &\hspace{-0.2cm}=\left\{ \begin{array}{l}
\sum\limits_{k=1}^{{K_{\max}^{\rm{b}}}} \mathcal{P}_{{\frac{\lambda_\mathrm{U}}{\lambda_\mathrm{S}}}}\left(k\right)+q_\mathrm{hit} \sum\limits_{k=K_{\max}^{\rm{b}}+1}^{{K_{\max}^{\rm{a}}}} \mathcal{P}_{{\frac{\lambda_\mathrm{U}}{\lambda_\mathrm{S}}}}\left(k\right),\; K_{\max}^{\rm{a}}\geq K_{\max}^{\rm{b}},
\\\sum\limits_{k=1}^{K_{\max}^{\rm{a}}} \mathcal{P}_{{\frac{\lambda_\mathrm{U}}{\lambda_\mathrm{S}}}}\left(k\right)+\left(1-q_\mathrm{hit}\right)
\\\qquad\quad~~~ \times \sum\limits_{k=K_{\max}^{\rm{a}}+1}^{{K_{\max}^{\rm{b}}}} \mathcal{P}_{{\frac{\lambda_\mathrm{U}}{\lambda_\mathrm{S}}}}\left(k\right)
,\; K_{\max}^{\rm{a}}<K_{\max}^{\rm{b}},
\end{array} \right.
\end{align}
where $K_{\max}^{\rm{a}}$ and $K_{\max}^{\rm{b}}$ are given by \eqref{K_max_eq_11} and \eqref{K_b_num}, respectively.

The SCD probability given in \eqref{SCD_massiveMIMO_backhaul_access} can be intuitively understood based on the fact that when the small cell load is light, UEs' requested contents can be successfully delivered whether they are cached or obtained from the core networks via massive MIMO backhaul. However, after a critical value of cell load, UEs can only obtain their requested contents that are cached by the SBSs or via backhaul, which depends on the maximum cell load in cached content delivery and self-backhauled content delivery cases.

\section{Content Placement, Cache Size and Latency}
In this section, we study the effects of content placement and cache size on the content delivery performance. Then, we evaluate the latency in such networks.

\subsection{Content Placement and Cache Size}
As shown in \eqref{SCD_massiveMIMO_backhaul_access}, hit probability plays an important role in content delivery. Since hit probability depends on
the cache size and content placement, SBSs with large storage capacity can cache more popular contents, to avoid
frequent backhaul and reduce backhaul cost and latency. Therefore, higher hit probability is meaningful to reduce the network's operational
and capital expenditures (OPEX, CAPEX). Given the SBS's cache size, different content placement strategies may result in various hit
probability, and caching the most popular contents (MPC) can achieve the highest hit probability, which is commonly-considered in the literature involving edge caching such as~\cite{E_Bastug_2014,S_H_Park_2016}. Therefore, we consider MPC caching
and analyze the appropriate cache size in such networks.  Considering the fact that for large $J$ with MPC caching,
$q_{\mathrm{hit}}={\sum\nolimits_{j= 1}^L {a_j }} \approx \left(\frac{L}{J}\right)^{1-\varsigma}$, we have
\begin{corollary}
Given $\frac{{{T_{\mathrm{th}}} - {T_1}}}{{{T_{\mathrm{th}}}}} \leq \frac{\eta }{ {1 - \eta } }$ (i.e., more time-frequency resources are allocated to the cached content delivery), the SCD probability is
\begin{align}\label{cache_size_factor}
{\Psi_{\mathrm{SCD}}}\left( {Q,T_\mathrm{th}} \right) \approx \sum\limits_{k=1}^{K_{\max}^{\rm{b}}} \mathcal{P}_{{\frac{\lambda_\mathrm{U}}{\lambda_\mathrm{S}}}}\left(k\right),
\end{align}
and it is an increasing function of the cache size, if the cache size $L\in \left[J\left(\left[1- \Xi_{\rm{b}}\frac{1-\epsilon}{\epsilon}\right]^{+}
\right)^{\frac{1}{1-\varsigma}},J\left(\frac{1}{2}\right)^{\frac{1}{1-\varsigma}}\right]$ and the minimum SBS density satisfies the condition given in \textbf{Corollary 6}; Given $\frac{{{T_{\mathrm{th}}} - {T_1}}}{{{T_{\mathrm{th}}}}} > \frac{\eta }{ {1 - \eta } }$, the SCD probability is
\begin{align}\label{cache_size_factor_cp}
{\Psi_{\mathrm{SCD}}}\left( {Q,T_\mathrm{th}} \right) \approx \sum\limits_{k=1}^{K_{\max}^{\rm{a}}} \mathcal{P}_{{\frac{\lambda_\mathrm{U}}{\lambda_\mathrm{S}}}}\left(k\right),
\end{align}
if $L\in \left[J\left(\frac{1}{2}\right)^{\frac{1}{1-\varsigma}},\left(\min\left\{ \Xi_{\rm{a}} \frac{1-\epsilon}{\epsilon},1\right\}\right)^{\frac{1}{1-\varsigma}}\right]$, and the minimum SBS density satisfies the condition given in \textbf{Corollary 4}.
\begin{proof}
See Appendix F.
\end{proof}

\end{corollary}
 The above corollary provides some important insights into the interplay between time-frequency resource allocation and cache size in cache-enabled dense cellular networks with massive MIMO backhaul, which plays a key role in the content delivery performance.

\subsection{Latency}
To evaluate the latency in such networks,  we consider the average delay for successfully obtaining the requested content in such networks. It should be noted that when the small cells are overloaded, UEs may suffer longer delay. There are many approaches to solve the overload issue such as  deploying enough small cells following the rule of \textbf{Corollary 2} and \textbf{Corollary 6} or advanced multi-antenna techniques. Moreover, it may be more lightly loaded in realistic small cell networks~\cite{Jeffrey_5G}.  For tractability, we assume that the load of a small cell will not exceed its maximum load $K_{\max}$. As suggested in \cite{J_Liu_2017}, the average delay for requesting a content from a typical small cell  can be expressed as
\begin{align}\label{average_delay_resource_allocation}
&\mathcal{D}=\sum\limits_{k=1}^{K_{\max}} \mathcal{P}_{{\frac{\lambda_\mathrm{U}}{\lambda_\mathrm{S}}}}\left(k\right) \bigg(q_\mathrm{hit} \frac{Q}{\mathbb{E}\left\{R_\mathrm{a}\right\}} \nonumber\\
&\qquad\qquad\qquad~~~~~~~+\left(1-q_\mathrm{hit}\right)
\Big(T_1+\frac{Q}{\mathbb{E}\left\{R_{\mathrm{a}^{'}}\right\}}\Big)\bigg),
\end{align}
where $T_1$ is the massive MIMO backhaul time detailed in Section III-B, and $\mathbb{E}\left\{R_\mathrm{a}\right\}$ and $\mathbb{E}\left\{R_{\mathrm{a}^{'}}\right\}$ are the average access rate of the cached  and self-backhauled content delivery, respectively, which are given by
\begin{align}\label{E_a_a_pi}
\left\{ \begin{array}{l}
\mathbb{E}\left\{R_\mathrm{a}\right\}=\int_{\rm{0}}^\infty  {\varphi \left( x,q_\mathrm{hit},\eta\right)}dx,
\\
\mathbb{E}\left\{R_{\mathrm{a}^{'}}\right\}=\int_{\rm{0}}^\infty  {\varphi \left( x,1-q_\mathrm{hit},1-\eta\right)} dx,
\end{array}\right.
\end{align}
where $\varphi \left(x,\theta_1,\theta_2\right)=\left(1+\theta_1\frac{2^{\frac{{k x}}{\theta_2 W }+1}-2}{\alpha_\mathrm{a}-2} \chi\left(k\right)\right)^{-1}$ with $\chi\left(k\right)={_2F_{1}}\left[1,1-\frac{2}{\alpha_\mathrm{a}};2-\frac{2}{\alpha_\mathrm{a}};
1-{2^{\frac{{k x}}{{\theta_2 W }}}} \right]$ is the complementary cumulative distribution function  of the $R_\mathrm{a}$ or $R_{\mathrm{a}^{'}}$, respectively,  which is obtained by using the approach in Appendix A.

Given the hit probability, i.e., the cache size is fixed, the spectrum fraction $\eta=\eta_o$ for meeting $\mathbb{E}\left\{R_\mathrm{a}\right\}=\mathbb{E}\left\{R_{\mathrm{a}^{'}}\right\}$ can be easily obtained by using one-dimension search, considering the fact that $\mathbb{E}\left\{R_\mathrm{a}\right\}-\mathbb{E}\left\{R_{\mathrm{a}^{'}}\right\}$ is an increasing function of $\eta$.
{\begin{corollary}
When $\eta<\eta_o$, the average delay of self-backhauled  content delivery could be lower than cached content delivery  if massive MIMO antennas meet
\begin{align}\label{delay_massiveMIMO_antennas}
N \geq  \left( {\frac{{{2^{\Theta \left( {{\eta _o}} \right)}} - 1}}{{{P_{\rm{b}}}\beta {e^{{{\overline \Delta }_1} - {{\overline \Delta }_2}}}}} + 1} \right){S_o} - \frac{1}{2}
\end{align}
with $\Theta \left( {{\eta _o}} \right) = \frac{{{{\left( {1 - {\eta _o}} \right)}^{ - 1}}{W^{ - 1}}\mathbb{E}\left\{ {{R_\mathrm{a}}} \right\}\mathbb{E}\left\{ {{R_{\mathrm{a}^{'}}}} \right\}}}{{\mathbb{E}\left\{ {{R_{{\mathrm{a}^{'}}}}} \right\} - \mathbb{E}\left\{ {{R_\mathrm{a}}} \right\}}}$, for a specified typical backhaul load $S_o$.
\end{corollary}

The proof of \textbf{Corollary 8} can be easily obtained by considering $T_1 \leq \frac{Q}{\mathbb{E}\left\{R_\mathrm{a}\right\}}-\frac{Q}{\mathbb{E}\left\{R_{\mathrm{a}^{'}}\right\}}$ for $\eta<\eta_o$ and \textbf{Corollary 3}. It is implied from \textbf{Corollary 8} that for the case of requesting non-cached contents, the average delay of the non-cached content delivery  via massive MIMO backhaul could be comparable to that of the cached content delivery, if the average access rate of cached content delivery is lower than that of self-backhauled content delivery.}

\section{Simulation Results}

\begin{table*}[tb]
\centering
\caption{Simulation Parameters}\label{table1}
\begin{tabular}{ |p{6cm}|p{3.5cm}| p{4.2cm}|}
 \hline
 Parameter & Symbol & Value \\
 \hline
 Pathloss exponent to UE & $\alpha_\mathrm{a}$ & 3.0\\
  \hline
 Pathloss exponent to SBS & $\alpha_\mathrm{b}$ & 2.6 \\
  \hline
Transmit power of MBS & $P_\mathrm{b}$ & 46 dBm\\
 \hline
Transmit power of SBS & $P_\mathrm{a}$ & 30 dBm\\
  \hline
Carrier frequency   & $f_c$  & 3.5 GHz \\
  \hline
 Frequency dependent constant value & $\beta$ & $\left(\frac{{{3 \times 10^8}}}{{4\pi {f_c}}}\right)^2$\\
  \hline
 System bandwidth   & $W$  & 100 MHz \\
  \hline
Noise power & $\sigma_\mathrm{a}^2,\;\sigma_\mathrm{b}^2,\;\sigma_{\mathrm{a}^{'}}^2$ & $-174+10 \times \log_{10}(\mathrm{Bandwidth})$ dBm\\
 \hline
Content library size & $J$ & $10^5$ \\
\hline
Zipf exponent & $\varsigma$ & $0.7$\\
\hline
\end{tabular}
\end{table*}
In this section, simulation results are presented to validate the prior analysis and further shed light on the effects of key system parameters including cell load, cache size, BS density, and massive MIMO antennas on the performance. The
basic simulation parameters are shown in Table I.

\subsection{Cached Content Delivery}
In this subsection, we illustrate the cell load, SCD probability, and minimum required SBS density when the requested content is cached at the associated SBS.

\begin{figure}[htbp]
\centering
\includegraphics[width=3.5 in,height=2.8 in]{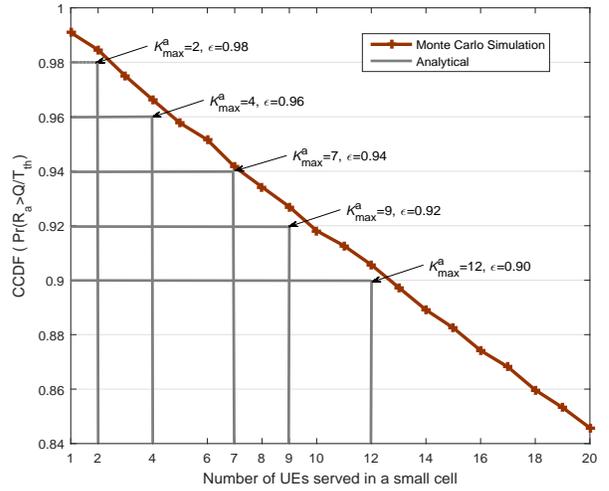}
\caption{The complementary cumulative distribution function (CCDF) of the $R_\mathrm{a}$: $\frac{Q}{T_{\mathrm{th}}}=1$ Mbps, $\lambda_\mathrm{U}=3\times 10^{-4}$ $\mathrm{m}^{-2}$, $\lambda_\mathrm{S}=10^{-4}$ $\mathrm{m}^{-2}$,  $\eta=0.5$, and Cache Size$=3\times 10^{3}$.}
\label{fig2}
\end{figure}
Fig. \ref{fig2} shows the complementary cumulative distribution function (CCDF) of the rate $R_\mathrm{a}$ for different number of UEs served in a small cell. The analytical maximum cell load $K_{\max}^{\mathrm{a}}$ for different CCDF thresholds are obtained from \eqref{K_max_eq_11}, which has a precise match with the Monte Carlo simulations. The CCDF is a decreasing function of number of UEs served in a small cell, since  resources allocated to each UE become less  when serving more UEs.

\begin{figure}[htbp]
\centering
\includegraphics[width=3.5 in,height=2.8 in]{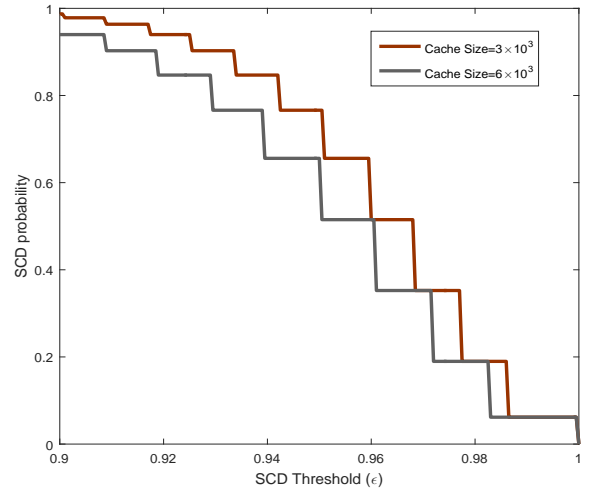}
\caption{The SCD probability: $\frac{Q}{T_{\mathrm{th}}}=1$ Mbps, $\lambda_\mathrm{U}=3\times 10^{-4}$ $\mathrm{m}^{-2}$, $\lambda_\mathrm{S}=10^{-4}$ $\mathrm{m}^{-2}$,  and $\eta=0.5$.}
\label{fig3}
\end{figure}

{ Fig. \ref{fig3} shows the SCD probability when the requested content is cached at the associated SBS, based on \textbf{Theorem 1} and Fig. \ref{fig2}. The stair-like curves are induced by the fact that the SCD probability given by \eqref{SCDP_cached} is a discrete function of maximum cell load.}  We see that for fixed cache size,  the SCD probability decreases when the system requires higher SCD threshold $\epsilon$, since higher $\epsilon$ reduces the level of maximum allowable cell load, as suggested in Fig. \ref{fig2}. Moreover, for a given $\epsilon$, the SCD probability decreases with increasing the cache size. The reason is that hit probability increases with increasing the cache size, i.e.,  UEs are more likely to obtain the requested contents cached by their associated SBSs, which results in more interference at the same frequency band and reduces the maximum allowable cell load.

\begin{figure}[htbp]
\centering
\includegraphics[width=3.5 in,height=2.8 in]{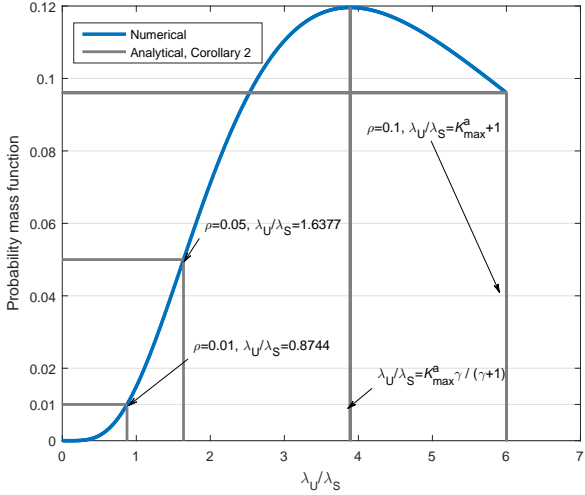}
\caption{The minimum required SBS density for avoiding overloading.}
\label{fig4}
\end{figure}

Fig. \ref{fig4} shows the  minimum required SBS density to avoid the overload issue given the UE density $\lambda_\mathrm{U}$. Without loss of generality, we assume that the maximum allowable load of a small cell is  $K_{\max}^\mathrm{a}=5$ in this figure (note that for specified system performance requirement, the maximum small cell load is obtained from \eqref{K_max_eq_11}, as illustrated in Fig. 2.).  The numerical result precisely matches with the analysis shown in \textbf{Corollary 2}. We see that when the probability that more than $K_{\max}^\mathrm{a}$ UEs need to be served in a small cell is not larger  than $\rho=0.1$, the minimum required SBS density satisfies $\frac{\lambda_\mathrm{U}}{\lambda_\mathrm{S}}=K_{\max}^\mathrm{a}+1=6$, as confirmed in \eqref{SBS_density_min_theo1}. When the system requires lower $\rho=0.1$ (i.e., lower overload probability.), the  density ratio $\frac{\lambda_\mathrm{U}}{\lambda_\mathrm{S}}$ in such networks decreases, which means that more SBSs need to be deployed.

\subsection{Massive MIMO Backhaul Transmission}
In this subsection, we focus on the massive MIMO backhaul achievable rate, which determines the amount of backhaul time when an SBS obtains the requested content from its associated MBS. Note that the macrocell load and minimum required MBS density have been studied in Section III-B,  which are similar to \textbf{Theorem 1} and \textbf{Corollary 2}, and numerical results can be easily obtained by following Figs. 2 and 4.
\begin{figure}[htbp]
\centering
\includegraphics[width=3.5 in,height=2.8 in]{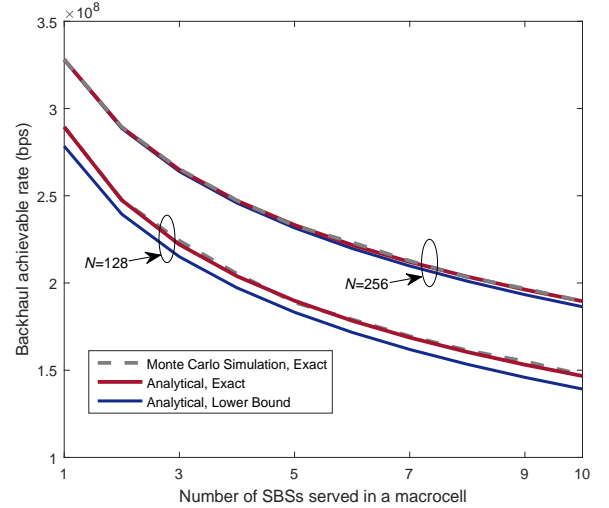}
\caption{Backhaul achievable rate: $\lambda_\mathrm{M}=10^{-5}$ $\mathrm{m}^{-2}$, $\eta=0.5$ and $r_\mathrm{b}=5$ m.}
\label{fig5}
\end{figure}

Fig. \ref{fig5} shows the backhaul achievable rate for different macrocell load and massive MIMO antennas. The analytical exact and lower-bound curves are obtained from \eqref{backhaul_tran_rate_exp} and \eqref{coro_3_eq_lowerbound}, respectively, which tightly matches with the simulated exact curves. We see that the backhaul achievable rate decreases when macrocell load increases, since each SBS will obtain less transmit power and array gains. Adding more massive MIMO antennas improves the achievable rate because of larger array gains.

\subsection{Latency}
In this subsection, we evaluate the average delay in two scenarios: 1)  The requested content is cached at the associated SBS;  and 2) the requested content is not cached and needs to be obtained via massive MIMO backhaul.

\begin{figure}[htbp]
\centering
\includegraphics[width=3.5 in,height=2.8 in]{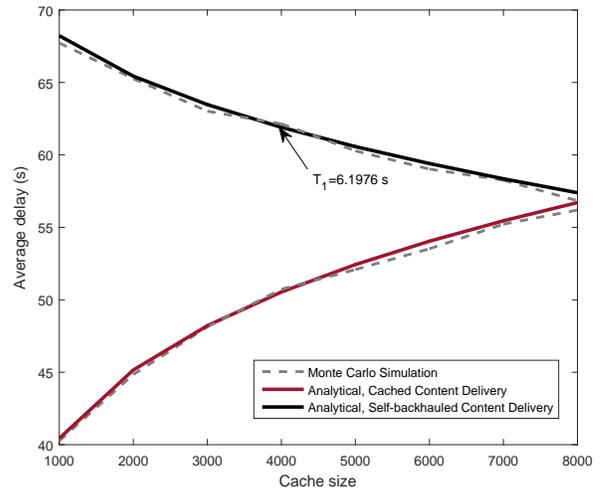}
\caption{Average delay: $Q=1$ Gbit, $\lambda_\mathrm{U}=3\times 10^{-4}$ $\mathrm{m}^{-2}$, $\lambda_\mathrm{S}=10^{-4}$ $\mathrm{m}^{-2}$,  $\lambda_\mathrm{M}=10^{-5}$ $\mathrm{m}^{-2}$, $N=128$, $S_o=10$, $K=5$, $\eta=0.45$, and $r_\mathrm{b}=5$ m.}
\label{fig6}
\end{figure}

Fig. \ref{fig6} shows the average delay for different cache size. The analytical curves are obtained based on the average rate given by \eqref{E_a_a_pi}. We see that the average delay for cached content delivery is lower than that of the self-backhauled content delivery. The average delay for cached content delivery increases with increasing the cache size. In contrast, the average delay for self-backhauled content delivery decreases with increasing the cache size. The reason is that larger cache size results in higher hit probability, and more SBSs can provide cached content delivery. This results in more inter-SBS interference over the frequency band allocated to the cached content delivery, and less inter-SBS interference over the frequency band allocated to the self-backhauled content delivery. In addition, the content delivery time of massive MIMO backhaul $T_1$ is much lower than that of the access.

\section{Conclusion}
{We have studied content delivery in cache-enabled HetNets with massive MIMO backhaul. In such networks, the successful content delivery probability involving cached content delivery and non-cached content delivery via massive MIMO backhaul was analyzed. The effects of hit probability, UE and SBS densities on the performance were addressed. Particularly, we provided the minimum required SBS and MBS densities for avoiding overloading. The derived results demonstrated that hit probability needs to be properly determined, in order to achieve successful content delivery. The interplay between cache size and time-frequency resource allocations was quantified from the perspective of successful content delivery probability. The latency was characterized in terms of average delay in this networks.  It was  proved that when UEs request non-cached contents, the average delay of the non-cached content delivery  could be comparable to that of the cached content delivery with the help of massive MIMO aided self-backhaul in some cases.}

\section*{Appendix A: Proof of {Theorem 1}}
\label{App:theo_1}
\renewcommand{\theequation}{A.\arabic{equation}}
\setcounter{equation}{0}

Based on  \eqref{access_Rate}, SCD probability is calculated as
\begin{align}\label{App_A_1}
{\Psi_{\mathrm{SCD}}^{\mathrm{a}}}\left( {Q,T_\mathrm{th}} \right)&=\Pr\left(R_{\rm{a}}\geq\frac{Q}{T_\mathrm{th}}\right)\nonumber\\
&=\mathbb{E}_K\Bigg\{\underbrace{\Pr\Big(R_{\rm{a}}\geq\frac{Q}{T_\mathrm{th}}|K=k\Big)}_{\Lambda\left(k\right)}\Bigg\} \nonumber\\
&=\sum\limits_{k =1} \mathcal{P}_{\frac{\lambda_\mathrm{U}}{\lambda_\mathrm{S}}}\left(k\right) \Lambda\left(k\right),
\end{align}
where $\mathcal{P}_{\frac{\lambda_\mathrm{U}}{\lambda_\mathrm{S}}}\left(k\right)$ is the probability mass function (PMF) of the number of other $k-1$ UEs (except typical UE) served by the typical SBS, and {$\Lambda\left(k\right)$ is the conditional SCD probability given $K=k$.} According to~\cite{Singh_May_2013}, $\mathcal{P}_{\frac{\lambda_\mathrm{U}}{\lambda_\mathrm{S}}}\left(k\right)$ can be calculated as
\begin{align}\label{App_A_2}
\mathcal{P}_{\frac{\lambda_\mathrm{U}}{\lambda_\mathrm{S}}}\left(k\right)&=\frac{\gamma^\gamma}{\left(k-1\right)!}\frac{\Gamma\left(k+\gamma\right)}{\Gamma\left(\gamma\right)}
\frac{\left(\frac{\lambda_\mathrm{U}}{\lambda_\mathrm{S} }\right)^{k-1}}{ \left(\gamma+\frac{\lambda_\mathrm{U}}{\lambda_\mathrm{S}}\right)^{k+\gamma}   },
\end{align}
where $\gamma=3.5$~\cite{J_S_Ferenc_2007}. Given $K=k$, $\Lambda\left(k\right)$ is calculated as
\begin{align}\label{App_A_3}
\Lambda\left(k\right)&=\Pr\left(R_{\rm{a}}\geq\frac{Q}{T_\mathrm{th}}\right) \nonumber\\
&=\mathbb{E}_{\left| {{X_o}} \right|}\left\{\Pr\left(\frac{{{P_{\mathrm{a}}}{h_o}L\left( {\left| {{X_o}} \right|} \right)}}{{I_{\mathrm{a}}+\sigma _{\mathrm{a}}^2}}{\rm{ \geq }}{2^{\frac{{kQ}}{{\eta W T_\mathrm{th}  }}}} - 1\right)\right\} \nonumber\\
&=\int_0^\infty \underbrace{\Pr\left(\frac{{{P_{\mathrm{a}}}{h_o}L\left( {x} \right)}}{{I_{\mathrm{a}}+\sigma _{\mathrm{a}}^2}}{\rm{ \geq }}{2^{\frac{{kQ}}{{\eta W T_\mathrm{th}  }}}} - 1\right)}_{\Upsilon_1\left(x\right)} f_{\left| {{X_o}} \right|}\left(x\right)   dx,
\end{align}
where $f_{\left| {{X_o}} \right|}\left(x\right)=2 \pi \lambda_\mathrm{S}  x \exp\left(-\pi \lambda_\mathrm{S} x^2\right)$ is the probability density function (PDF) of the distance between the typical UE and its associated SBS, and $\Upsilon_1\left(x\right)$ is the conditional SCD probability given $K=k$ and $\left| {{X_o}} \right|=x$. Considering the fact that dense cellular network is interference-limited in practice, the effect of noise power on the performance is negligible. As such, we can evaluate $\Upsilon_1\left(x\right)$ as
\begin{align}\label{Upsilon_1_cached}
&\Upsilon_1\left(x\right)=\mathbb{E}_{\Phi_\mathrm{S}^\mathrm{a}}\left\{\exp\left(-\frac{{2^{\frac{{kQ}}{{\eta W T_\mathrm{th}  }}}} - 1}{P_{\mathrm{a}} L\left(x\right)}I_{\mathrm{a}} \right)   \right\} \nonumber\\
&\mathop =\limits^{\left( a \right)}  \exp \left(-2 \pi \lambda_\mathrm{S} q_{\mathrm{hit}} \int_x^\infty \frac{\left({2^{\frac{{kQ}}{{\eta W T_\mathrm{th}  }}}} - 1\right)x^{\alpha_\mathrm{a}}r^{1-\alpha_\mathrm{a}}}{1+ \left({2^{\frac{{kQ}}{{\eta W T_\mathrm{th}  }}}} - 1\right)x^{\alpha_\mathrm{a}} r^{-\alpha_\mathrm{a}}}   dr  \right) \nonumber\\
&=\exp\bigg(-2 \pi \lambda_\mathrm{S} q_{\mathrm{hit}}\frac{x^{2}}{\alpha_\mathrm{a}-2} \left({2^{\frac{{kQ}}{{\eta W T_\mathrm{th}  }}}} - 1\right) \times \nonumber\\
&\qquad\quad  {_2F_{1}}\left[1,1-\frac{2}{\alpha_\mathrm{a}};2-\frac{2}{\alpha_\mathrm{a}};
1-{2^{\frac{{kQ}}{{\eta W T_\mathrm{th}  }}}} \right] \bigg),
\end{align}
where step (a) is obtained by using the generating functional of the PPP~\cite{Haenggi2009}. By substituting \eqref{Upsilon_1_cached} into \eqref{App_A_3},
$\Lambda\left(k\right)$ can be derived in closed-form  as
\begin{align}\label{Appendix_A_Lambda}
\Lambda\left(k\right)=\frac{1}{1+ q_{\mathrm{hit}}\frac{2^{\frac{{kQ}}{\eta W T_\mathrm{th}}+1}-2}{\alpha_\mathrm{a}-2} \chi_k^{\mathrm{a}}\left(k\right)},
\end{align}
where $\chi_k^{\mathrm{a}}\left(k\right)={_2F_{1}}\left[1,1-\frac{2}{\alpha_\mathrm{a}};2-\frac{2}{\alpha_\mathrm{a}};
1-{2^{\frac{{k Q}}{{\eta W T_\mathrm{th}  }}}} \right]$. Based on \eqref{Appendix_A_Lambda}, the maximum load $K_{\max}^{\rm{a}}$ of a typical small cell is given by
\begin{align}\label{k_SBS_load}
\left.\Lambda\left(k\right)\right|_{k=K_{\max}^{\rm{a}}}=\epsilon,
\end{align}
where $\epsilon$ is the threshold that SCD occurs when $\Lambda\left(k\right)\geq\epsilon$. {Although the closed-form solution with respect to (w.r.t.) $k={K_{\max }^{\rm{a}}}$ of \eqref{k_SBS_load} is unfeasible, it can be efficiently obtained by using one-dimension search as detailed in \textbf{Algorithm 1} due to the fact that $\Lambda\left(k\right)$ is a decreasing function of $k$.} The SCD probability in \eqref{App_A_1} is rewritten as
\begin{align}\label{SCD_SBS_App}
{\Psi_{\mathrm{SCD}}^{\mathrm{a}}}\left( {Q,T_\mathrm{th}} \right)&=\sum\limits_{k = 1}^{{K_{\max }^{\rm{a}}}} {\mathcal{P}_{\frac{\lambda_\mathrm{U}}{\lambda_\mathrm{S}}}\left(k\right)},
\end{align}
where ${\mathcal{P}_{\frac{\lambda_\mathrm{U}}{\lambda_\mathrm{S}}}\left(k\right)}$ and $K_{\max }^{\rm{a}}$  are defined by \eqref{App_A_2} and \eqref{k_SBS_load}, respectively, and the proof of \textbf{Theorem 1} is completed.

\section*{Appendix B: Proof of {Corollary 2}}
\label{App:theo_1}
\renewcommand{\theequation}{B.\arabic{equation}}
\setcounter{equation}{0}

After obtaining $K_{\max}^{\rm{a}}$, we can  find out how many small cells are sufficient to serve a specified scale of UEs $\lambda_\mathrm{U}$, since serving larger than $K_{\max}^{\rm{a}}$ UEs in a small cell cannot achieve SCD.  Assuming that $\mathcal{P}_{\frac{\lambda_\mathrm{U}}{\lambda_\mathrm{S}}}\left(K_{\max}^{\rm{a}}+1\right)=\rho$ with arbitrary small $\rho>0$, we need to guarantee  $\mathcal{P}_{\frac{\lambda_\mathrm{U}}{\lambda_\mathrm{S}}}\left(k\right) \leq \rho,\forall k>K_{\max}^{\rm{a}}$, in order to avoid content delivery failure resulting from overloading. Let
\begin{align}\label{k_k+1_ratio}
\frac{\mathcal{P}_{\frac{\lambda_\mathrm{U}}{\lambda_\mathrm{S}}}\left(k+1\right)}{\mathcal{P}_{\frac{\lambda_\mathrm{U}}{\lambda_\mathrm{S}}}\left(k\right)}=
\left(1+\frac{\gamma}{k}\right)\frac{\frac{\lambda_\mathrm{U}}{\lambda_\mathrm{S}}}{\gamma+\frac{\lambda_\mathrm{U}}{\lambda_\mathrm{S}}}
\leq 1, k \geq K_{\max}^{\rm{a}}+1.
\end{align}
We can intuitively interpret \eqref{k_k+1_ratio} based on the fact that given the maximum load $K_{\max}^{\rm{a}}$, the probability that serving more than $K_{\max}^{\rm{a}}$ UEs should be lower when adding more UEs. From \eqref{k_k+1_ratio}, we get $\frac{\lambda_\mathrm{U}}{\lambda_\mathrm{S}} \leq K_{\max}^{\rm{a}}+1$ such that $\mathcal{P}_{\frac{\lambda_\mathrm{U}}{\lambda_\mathrm{S}}}\left(k\right) \leq \rho,\forall k>K_{\max}^{\rm{a}}$. Then, we need to solve $\mathcal{P}_{\frac{\lambda_\mathrm{U}}{\lambda_\mathrm{S}}}\left(K_{\max}^{\rm{a}}+1\right)=\rho$ w.r.t. $\frac{\lambda_\mathrm{U}}{\lambda_\mathrm{S}}$ under the constraint $\frac{\lambda_\mathrm{U}}{\lambda_\mathrm{S}} \leq K_{\max}^{\rm{a}}+1$. The first-order partial derivative of $\mathcal{P}_{\frac{\lambda_\mathrm{U}}{\lambda_\mathrm{S}}}\left(k\right)$ w.r.t. $\frac{\lambda_\mathrm{U}}{\lambda_\mathrm{S}}$ is
\begin{align}\label{first_order_ratio}
\frac{{\partial {\mathcal{P}_{{\frac{\lambda_\mathrm{U}}{\lambda_\mathrm{S}}}}}}\left(k\right)}{\partial {\frac{\lambda_\mathrm{U}}{\lambda_\mathrm{S}}} }&=\frac{\gamma^\gamma \Gamma\left(k+\gamma\right)}{\left(k-1\right)!\Gamma\left(\gamma\right)}\left(\frac{\lambda_\mathrm{U}}{\lambda_\mathrm{S} }\right)^{k-2}\left(\gamma+\frac{\lambda_\mathrm{U}}{\lambda_\mathrm{S}}\right)^{-\left(k+\gamma+1\right)} \nonumber\\
&\quad \times \left(\left(k-1\right)\gamma-\left(\gamma+1\right)\frac{\lambda_\mathrm{U}}{\lambda_\mathrm{S}}\right).
\end{align}
{From \eqref{first_order_ratio}, we see that for $k=K_{\max}^{\rm{a}}+1$, $\frac{{\partial {\mathcal{P}_{{\frac{\lambda_\mathrm{U}}{\lambda_\mathrm{S}}}}}}}{\partial {\frac{\lambda_\mathrm{U}}{\lambda_\mathrm{S}}} } \geq 0$ as $\frac{\lambda_\mathrm{U}}{\lambda_\mathrm{S}}\in\Big(0,\frac{K_{\max}^{\rm{a}}\gamma}{\gamma+1}\Big] $, and $\frac{{\partial {\mathcal{P}_{{\frac{\lambda_\mathrm{U}}{\lambda_\mathrm{S}}}}}}}{\partial {\frac{\lambda_\mathrm{U}}{\lambda_\mathrm{S}}} }< 0$ as $\frac{\lambda_\mathrm{U}}{\lambda_\mathrm{S}}\in\Big(\frac{K_{\max}^{\rm{a}}\gamma}{\gamma+1}, K_{\max}^{\rm{a}}+1\Big]$. Therefore, the minimum required density of SBSs satisfies
\begin{align}\label{SBS_density_min}
\frac{\lambda_\mathrm{U}}{\lambda_\mathrm{S}}=\left\{ \begin{array}{l}
\left(K_{\max}^{\rm{a}}+1\right),\; if \;{\mathcal{P}_{\frac{\lambda_\mathrm{U}}{\lambda_\mathrm{S}}=K_{\max}^{\rm{a}}+1}\left(K_{\max}^{\rm{a}}+1\right)} \leq \rho,
\\
\mu_\mathrm{a},\; if \;{\mathcal{P}_{\frac{\lambda_\mathrm{U}}{\lambda_\mathrm{S}}=K_{\max}^{\rm{a}}+1}\left(K_{\max}^{\rm{a}}+1\right)}> \rho,
\end{array} \right.
\end{align}
where $\mu_\mathrm{a} \in \Big(0,\frac{K_{\max}^{\rm{a}}\gamma}{\gamma+1}\Big]$ is the solution of $\mathcal{P}_{\frac{\lambda_\mathrm{U}}{\lambda_\mathrm{S}}=\mu_\mathrm{a}}\left(K_{\max}^{\rm{a}}+1\right)=\rho$, and can be easily obtained by using one-dimension search approach, since $\mathcal{P}_{\frac{\lambda_\mathrm{U}}{\lambda_\mathrm{S}}=\mu_\mathrm{a}}\left(K_{\max}^{\rm{a}}+1\right)$ is an increasing function of $\mu_\mathrm{a}$ as $\mu_\mathrm{a} \in \Big(0,\frac{K_{\max}^{\rm{a}}\gamma}{\gamma+1}\Big]$.} Thus, we obtain the minimum required SBS density, in order to avoid overloading.

\section*{Appendix C: Detailed derivation of \eqref{backhaul_tran_rate_exp}}
\label{App:theo_1}
\renewcommand{\theequation}{C.\arabic{equation}}
\setcounter{equation}{0}
Since the typical SBS is associated with the nearest MBS, the PDF of the typical communication distance is 
\begin{align}\label{nearest_backhaul_dis}
f_{\left|Y_o\right|}\left(y\right)=\frac{2 \pi \lambda_\mathrm{M} y}{e^{-\pi\lambda_\mathrm{M} r_\mathrm{b}^2}} e^{-\pi \lambda_\mathrm{M} y^2},\;y \geq r_\mathrm{b},
\end{align}
where $r_\mathrm{b}$ is the minimum distance between the typical MBS and its associated SBS.
According to \eqref{rate_backhaul} and \cite{J_JOSE_2011,J_Hoydis_2013}, the achievable transmission rate  can be written as
\begin{align}\label{App_B_x_1}
&\overline{R}_\mathrm{b}={\left(1-\eta\right) W} \mathbb{E}_{\left|Y_o\right|}\left\{\log_2\Bigg(1+\frac{\frac{{{P_\mathrm{b}}}}{S_o} \Xi_1}{{\frac{{{P_\mathrm{b}}}}{S_o}\Xi_2+\Xi_3+\sigma_\mathrm{b}^2}}\Bigg)\right\} \nonumber\\
&={\left(1-\eta\right) W}  \int_{r_\mathrm{b}}^\infty C_\mathrm{b}\left(y\right) f_{\left|Y_o\right|}\left(y\right) dy,
\end{align}
where $C_\mathrm{b}\left(y\right)=\log_2\Bigg(1+\frac{\frac{{{P_\mathrm{b}}}}{S_o} \Xi_1\left(y\right)}{{\frac{{{P_\mathrm{b}}}}{S_o}\Xi_2\left(y\right)+\Xi_3\left(y\right)+\sigma_\mathrm{b}^2}}\Bigg)$ with $\Xi_1(y)=L\left( y \right)\left(\mathbb{E}\left\{\sqrt{{g_o}  }\right\}\right)^2$, $\Xi_2\left(y\right)=L\left( y \right)\mathrm{var}\left\{\sqrt{{g_o}}\right\}$,\footnote{$\mathrm{var}\left\{\cdot\right\}$ is the variance operator.}  and $\Xi_3\left(y\right)=\mathbb{E}_{\left|Y_o\right|=y}\left\{I_\mathrm{b}\right\}$.

 We first calculate $\Xi_1$ as
\begin{align}\label{Xi_1_App_B}
\Xi_1\left(y\right)&=L\left( y \right) \left(\int_{\rm{0}}^\infty  {\sqrt x } \frac{{{x^{N - S_o}}{e^{ - x}}}}{{\Gamma \left( {N - S_o+ 1} \right)}}dx\right)^2  \nonumber\\
&=L\left( y \right)\left(\frac{ \Gamma \left( {N - S_o+ \frac{3}{2}} \right)}{\Gamma \left( {N - S_o+ 1} \right)}\right)^2.
\end{align}
Then, $\Xi_2$ is given by
\begin{align}\label{Xi_2_App_B}
\Xi_2\left(y\right)&=L\left( y \right)  \mathbb{E}\left\{g_o \right\}-\Xi_1 =(N-S+1)L\left( y \right)-\Xi_1.
\end{align}
By using the Campbell's theorem~\cite{Baccelli2009}, $\Xi_3$ is obtained as
\begin{align}\label{Xi_3_app_B}
\Xi_3\left(y\right)&=\frac{P_\mathrm{b}}{S_j}\mathbb{E}\left\{g_j\right\} 2 \pi \lambda_\mathrm{M} \beta  \int_y^\infty t^{1-\alpha_\mathrm{b}} dt \nonumber\\
&=P_\mathrm{b} 2 \pi \lambda_\mathrm{M} \beta \frac{y^{2-\alpha_\mathrm{b}} }{\alpha_\mathrm{b}-2}.
\end{align}
By substituting \eqref{Xi_1_App_B}, \eqref{Xi_2_App_B} and \eqref{Xi_3_app_B} into \eqref{App_B_x_1}, we obtain \eqref{backhaul_tran_rate_exp}.

\section*{Appendix D: Proof of {Corollary 3}}
\label{App:theo_1}
\renewcommand{\theequation}{D.\arabic{equation}}
\setcounter{equation}{0}

According to the Stirling's formula, i.e., $\Gamma\left(x+1\right)\approx \left(\frac{x}{e}\right)^x\sqrt{2\pi x}$ as $x\rightarrow \infty$, we have
\begin{align}\label{asym_N_Xi_App}
{\Xi}_1\left(y\right) & \approx L\left( y \right) \left( \frac{\left(\frac{N-S_o+\frac{1}{2}}{e}\right)^{N-S_o+\frac{1}{2}} \sqrt{2\pi\left(N-S_o+\frac{1}{2}\right)}}{ \left(\frac{N-S}{e}\right)^{N-S_o} \sqrt{2\pi \left(N-S_o\right)}}\right)^2 \nonumber\\
&\approx L\left( y \right)\frac{N-S_o+\frac{1}{2}}{e} \left(1+\frac{1}{2\left(N-S_o\right)}\right)^{2\left(N-S_o\right)} \nonumber\\
&\mathop  \approx \limits^{\left( a \right)}  \left(N-S_o+\frac{1}{2}\right) L\left( y \right),
\end{align}
when the number of antennas at the MBS grows large. Note that step $\left(a\right)$ is obtained by the fact that $\left(1+\frac{1}{x}\right)^x \approx e$ as $x\rightarrow \infty$. Thus, ${\Xi}_2\left(y\right)=\frac{L\left( y \right)}{2}$. By using Jensen's inequality~\cite{Lifeng_Massive_MIMO}, we derive a tight lower-bound
on the achievable transmission rate \eqref{App_B_x_1} as
\begin{align}\label{lower_bound_backH}
&\overline{R}_\mathrm{b}^{\mathrm{Low}}={\left(1-\eta\right) W} \log_2\left(1+\frac{P_\mathrm{b}}{S_o}e^{\Delta_1-\Delta_2} \right),
\end{align}
where
 \begin{equation}\label{SINR_part2_mmwave}
 \left\{\begin{aligned}
 & \Delta_1=\mathbb{E}_{\left|Y_o\right|}\left\{\ln\Xi_1\right\}, \\
 & \Delta_2=\mathbb{E}_{\left|Y_o\right|}\left\{\ln\left(\frac{{P_\mathrm{b}}}{S_o}\Xi_2+\Xi_3+\sigma_\mathrm{b}^2\right) \right\}. \end{aligned}\right.
 \end{equation}
For large $N$, based on \eqref{asym_N_Xi_App}, $\Delta_1$ can be asymptotically derived as
\begin{align}\label{Delta_11}
&\Delta_1\approx \ln\left(N-S_o+\frac{1}{2}\right)+\mathbb{E}\left\{\ln L\left(y\right)\right\} \nonumber\\
&=\ln\left(N-S_o+\frac{1}{2}\right) +\ln\left(\beta\right)\nonumber\\
&\qquad\underbrace{-\frac{\alpha_\mathrm{b}}{e^{-\pi\lambda_\mathrm{M} r_\mathrm{b}^2}}\Big(-\frac{Ei\big(-r_\mathrm{b}^2 \pi \lambda_\mathrm{M} \big)}{2}+e^{-r_\mathrm{b}^2 \pi \lambda_\mathrm{M}}\ln r_\mathrm{b} \Big)}_{\overline{\Delta}_1},
\end{align}
where $Ei\left(z\right)$ is the exponential integral given by $Ei\left(z\right)=-\int_{-z}^\infty \frac{e^{-t}}{t} dt$. Then, $\Delta_2$ can be asymptotically calculated as
\begin{align}\label{Delta_2_asympt}
\Delta_2&= \int_{r_\mathrm{b}}^\infty \ln\left(\frac{{P_\mathrm{b}}}{S_o}\Xi_2\left(y\right)+\Xi_3\left(y\right)+\sigma_\mathrm{b}^2\right)f_{\left|Y_o\right|}\left(y\right) dy            \nonumber\\
&\approx \int_{r_\mathrm{b}}^\infty \ln\left(\frac{P_\mathrm{b}\beta}{2S_o}y^{-r_\mathrm{b}}+P_\mathrm{b} 2 \pi \lambda_\mathrm{M} \beta \frac{y^{2-\alpha_\mathrm{b}} }{\alpha_\mathrm{b}-2}+\sigma_\mathrm{b}^2\right) \nonumber\\
&\;\;\;\underbrace{\qquad\;\; \times \frac{2 \pi \lambda_\mathrm{M} y}{e^{-\pi\lambda_\mathrm{M} r_\mathrm{b}^2}} e^{-\pi \lambda_\mathrm{M} y^2} dy}_{\overline{\Delta}_2}.
\end{align}
{Substituting \eqref{Delta_11} and \eqref{Delta_2_asympt} into \eqref{lower_bound_backH}, we obtain \eqref{coro_3_eq_lowerbound}. }

Considering the fact that $T_1=\frac{Q}{\overline{R}_\mathrm{b}} \leq \frac{Q}{\overline{R}_\mathrm{b}^{\mathrm{Low}}}$, we obtain $T_1 \leq \frac{Q}{\left(1-\eta\right) W}\left(\log_2\left(1+\frac{P_\mathrm{b}\beta\left(N-S_o+\frac{1}{2}\right)}{S_o}e^{\overline{\Delta}_1-\overline{\Delta}_2} \right)\right)^{-1}$, which confirms the \textbf{Corollary 3}.

\section*{Appendix E: Proof of Theorem 2}
\label{App:theo_2}
\renewcommand{\theequation}{E.\arabic{equation}}
\setcounter{equation}{0}

Based on \eqref{AL_backhaul}, SCD probability is given by
\begin{align}\label{App_C_1}
{\Psi_{\mathrm{SCD}}^\mathrm{b}}\left( {Q,T_\mathrm{th}} \right)&=\Pr\left(R_{\mathrm{a}^{'}}>\frac{Q}{T_\mathrm{th}-T_1}\right)\nonumber\\
&=\sum\limits_{k \ge 1} \mathcal{P}_{\frac{\lambda_\mathrm{U}}{\lambda_\mathrm{S}}}\left(k\right) \Lambda_k^{\mathrm{b}},
\end{align}
where $\mathcal{P}_{\frac{\lambda_\mathrm{U}}{\lambda_\mathrm{S}}}\left(k\right)$ is given by \eqref{App_A_2}, and $\Lambda_k^{\mathrm{b}}$ is the conditional SCD probability given $K=k$. Similar to \eqref{App_A_3}, $\Lambda_k^{\mathrm{b}}$ is calculated as
\begin{align}\label{SCD_lambda_k_backhaul}
&\Lambda_k^{\mathrm{b}}=\Pr\left(\frac{{{P_{\mathrm{a}^{'}}}{h_o}L\left( {\left| {{X_o}} \right|} \right)}}{{I_{\mathrm{a}^{'}}+
\sigma _{\mathrm{a}^{'}}^2}}{\rm{ > }}{2^{\frac{{kQ}}{{\left(1-\eta\right) W \left(T_\mathrm{th}-T_1\right)  }}}} - 1\right) \nonumber\\
&=\frac{1}{1+ \left(1-q_{\mathrm{hit}}\right)\frac{2^{\frac{{kQ}}{{\left(1-\eta\right) W \left(T_\mathrm{th}-T_1\right)  }}+1}-2}{\alpha_\mathrm{a}-2}\chi_k^{\mathrm{b}}},
\end{align}
where $\chi_k^{\mathrm{b}}={_2F_{1}}\left[1,1-\frac{2}{\alpha_\mathrm{a}};2-\frac{2}{\alpha_\mathrm{a}};
1-{2^{\frac{{kQ}}{{\left(1-\eta\right) W \left(T_\mathrm{th}-T_1\right)  }}}} \right]$. Like \eqref{k_SBS_load}, the maximum load $K_{\max}^{\rm{b}}$ of a typical small cell is the solution of $\left.\Lambda\left(k\right)\right|_{k=K_{\max}^{\rm{b}}}=\epsilon$. Then, the SCD probability is obtained as \eqref{SCDP_backhaul_pro}.

\section*{Appendix F: Proof of Corollary 7}
\label{App:theo_2}
\renewcommand{\theequation}{F.\arabic{equation}}
\setcounter{equation}{0}
{Based on \eqref{K_max_eq_11} and \eqref{K_b_num}, we see that $K_{\max}^{\rm{a}} \geq K_{\max}^{\rm{b}}$ if $\frac{{{T_{\mathrm{th}}} - {T_1}}}{{{T_{\mathrm{th}}}}} \leq \frac{\eta }{ {1 - \eta } }$ and $q_\mathrm{hit} \leq \frac{1}{2}$. In this case, UE's requested contents are more likely to be delivered via massive MIMO self-backhaul. As such, based on \eqref{SCD_massiveMIMO_backhaul_access},  the SCD probability can be approximated as ${\Psi_{\mathrm{SCD}}}\left( {Q,T_\mathrm{th}} \right) \approx \sum\limits_{k=1}^{K_{\max}^{\rm{b}}} \mathcal{P}_{{\frac{\lambda_\mathrm{U}}{\lambda_\mathrm{S}}}}\left(k\right)$ given in \eqref{cache_size_factor}. Considering the fact that $q_{\mathrm{hit}}= \left(\frac{L}{J}\right)^{1-\varsigma}\geq \left[1- \Xi_{\rm{b}}\frac{1-\epsilon}{\epsilon}\right]^{+}$ in \textbf{Corollary 5} and $q_\mathrm{hit} =\left(\frac{L}{J}\right)^{1-\varsigma}\leq \frac{1}{2}$, the corresponding cache size is obtained as $L\in \left[J\left(\left[1- \Xi_{\rm{b}}\frac{1-\epsilon}{\epsilon}\right]^{+} \right)^{\frac{1}{1-\varsigma}},J\left(\frac{1}{2}\right)^{\frac{1}{1-\varsigma}}\right]$. Moreover, increasing the cache size boosts the hit probability, and thus enhances the maximum cell load $K_{\max}^{\rm{b}}$ for self-backhauled content delivery. The reason is that more SBSs can deliver the cached contents, which reduces inter-cell interference in self-backhauled content delivery.

Likewise, $K_{\max}^{\rm{a}} <  K_{\max}^{\rm{b}}$ if $\frac{{{T_{\mathrm{th}}} - {T_1}}}{{{T_{\mathrm{th}}}}} > \frac{\eta }{ {1 - \eta } }$ and $q_\mathrm{hit} > \frac{1}{2}$, and we can obtain \eqref{cache_size_factor_cp} and the corresponding cache size $L\in \left[J\left(\frac{1}{2}\right)^{\frac{1}{1-\varsigma}},\left(\min\left\{ \Xi_{\rm{a}} \frac{1-\epsilon}{\epsilon},1\right\}\right)^{\frac{1}{1-\varsigma}}\right]$. }

\bibliographystyle{IEEEtran}

\end{document}